\documentclass[11pt,a4paper]{article}

\newcommand{\hideshow}[1]{{\mbox{}}}

\usepackage{hyperref}
\usepackage{breakurl}
\usepackage{marginnote}
\usepackage{graphicx}
\usepackage{epstopdf}
\usepackage{multicol}
\usepackage{amsmath}
\usepackage{proof}
\usepackage{stmaryrd}
\usepackage{latexsym}
\usepackage{amsfonts}
\usepackage{bussproofs}
\usepackage{authblk}
\usepackage[all]{xy}
\usepackage{prettyref}
\usepackage{tabularx}
\usepackage{listings}
\usepackage{color}
\usepackage{wrapfig}
\usepackage[utf8]{inputenc}
\usepackage[normalem]{ulem}
\usepackage{amsthm}
\usepackage{txfonts}

\newtheorem{definition}{Definition}
\newtheorem{proposition}{Proposition}
\newtheorem{lemma}{Lemma}
\newtheorem{counterexample}{Counterexample}
\newtheorem{notation}{Notation}
\newtheorem{remark}{Remark}
\newtheorem{corollary}{Corollary}
\newtheorem{theorem}{Theorem}

\definecolor{mygreen}{rgb}{0,0.6,0}
\definecolor{mygray}{rgb}{0.5,0.5,0.5}
\definecolor{mymauve}{rgb}{0.58,0,0.82}

\newcommand{\p}[1]{\mbox{$[\![#1]\!]$}}
\newcommand{\pI}[1]{\mbox{$[\![#1]\!]^{\mathcal I}$}}

\newcommand{\reguno}[2]
  {
  $  \ \ \frac{\textstyle #1}{\textstyle #2} $
  }

\newcommand{\E}[2]{\ensuremath{{\epsilon}}}

\newcommand {\eg}        {{\textit{e}.\textit{g}.}}

\newcommand {\ie}        {{\textit{i}.\textit{e}.}}

\newtheorem{thermo}{Theorem}

\newtheorem{example}{Example}

\begin{document}

\title{Principal Types as Partial Involutions\thanks{A version of this work is currently under review for publication in the journal ``Mathematical Structures in Computer Science'' (\burl{https://www.cambridge.org/core/journals/mathematical-structures-in-computer-science}).}}

\author{Furio Honsell, Marina Lenisa, Ivan Scagnetto}
\date{Universit\`{a} di Udine, Italy\\
  \texttt{\{furio.honsell, marina.lenisa, ivan.scagnetto\}@uniud.it}}

\maketitle

\begin{abstract}
We show that the {\em principal types} of the closed terms of the affine fragment of $\lambda$-calculus, with respect to a simple type discipline, are structurally isomorphic to  their interpretations, as {\em partial involutions}, in a natural {\em Geometry of Interaction} model  \`a la Abramsky.
This permits to explain in elementary terms the somewhat awkward notion of {\em linear application} arising in {\em Geometry of Interaction}, simply as the {\em resolution} between principal types using an alternate unification algorithm. As a consequence,
we provide an answer, for the purely affine fragment,  to the open problem raised by Abramsky of characterising those partial involutions which are denotations of combinatory terms.
\end{abstract}


\section{Introduction}
The purpose of this paper is to provide an explanation of \emph{Geometry of Interaction (GoI)} semantics, first introduced in the seminal papers by \cite{Girard1,Girard2}, in terms of the more ordinary notions of \emph{type assignment systems}, \emph{principal types} and \emph{unification}   thereof. We provide, in fact, an {\em elementary alternative route} to the results based on Girard's work, where proof nets and $C^*$-algebras are used, as in \eg\  \cite{Baillot01}. For simplicity, we focus on the affine fragment of untyped  $\lambda$-calculus, where only variables which occur \emph{at most once} can be abstracted, and we provide full proofs of the results in \cite{CHLS18} (see also  \cite{CDGHLS19}). Our approach is completely elementary and could be extended to the full $\lambda$-calculus (as outlined in \eg\ \cite{HLS20}), and to all its complexity-oriented subcalculi; but ``dotting all i's and crossing all t's', even for this simpler fragment is not immediate.

Our line of reasoning stems from realizing a natural {\em structural analogy} between the {\em principal type} of a $\lambda$-term, actually a combinatory logic term, in the {\em simple types discipline}, and its interpretation as a  {\em partial involution} in a GoI model \emph{\`a la} \cite{Abr05}. Namely, labelled paths used in denoting occurrences of type-variables in types and  moves in the strategy denoting the term as a partial involution are isomorphic. This derives from the connection between paths to type variables in principal types and paths to term variables in normal forms. Such kind of relations have always been very fruitful in $\lambda$-calculus, starting from the well-known connection between Levy labels and
types,~\cite{L78}, and then between paths and
labels.
This analogy permits to explain simply as {\em resolution between principal types} the somewhat involved and puzzling, {\em particle style linear application} between involutions (see \cite{Abr96,Abr97,AHS02,AL05,Abr05}),   namely the notion of application used  in GoI and games semantics, based on Girard's Execution Formula.
We prove that GoI linear application amounts to {\em unifying} the left-hand side of the principal type of the operator with the principal type  of the operand, and applying the resulting substitution to the
 right-hand side of the operator. The jist of GoI is that it implements a novel \emph{bottom-up, variable-occurrence oriented}  unification algorithm alternate to the traditional one, which is top-down. This was pointed out also in \eg\ \cite {H23}.

In our view, this analysis sheds new light on the nature of Game Semantics itself and unveils the equivalence of three conceptually independent accounts  of application in $\lambda$-calculus: {\em $\beta$-reduction}, the GoI application of involutions based on symmetric feedback/Girard's {\em Execution Formula}, and {\em unification} of principal types.

As an interesting by-product, this paper provides an answer, for the  affine part, to the open problem raised in \cite{Abr05} of characterising the partial involutions
which are denotations of  combinator terms or, equivalently, arising from the bi-orthogonal pattern matching automata, which are there introduced. Namely, these are the partial involutions which encode the principal types of the combinator terms, \ie\ the theorems of minimal affine logic.

More precisely, we proceed as follows.

We introduce a type system for assigning principal types to affine $\lambda$-terms, where application is explained in terms of  resolution of principal types. We show that this system satisfies a number of interesting properties, such as a restricted form of {\em subject conversion}, and {\em uniqueness} of principal types. Moreover, we prove that principal types are {\em binary}, \ie\ types where each variable appears at most twice. Many of these results have straightforward  proofs, some of which had appeared before in the literature, see \eg\ \cite{H89}, however they had never been connected to GoI or put to use as we do in this paper.

We show then that binary types induce {\em immediately} an algebra of partial involutions ${\mathcal I}$, \ie\ history-free strategies on a suitable set of moves, in the style of \cite{Abr05}, where application amounts to GoI linear application. ${\mathcal I}$ turns out to be an {\em affine combinatory algebra} (a {\bf BCK}-algebra), \ie\ a model of affine combinatory logic. Moreover, we show that it is a {\em $\lambda$-algebra}, \ie\ a model of $\lambda$-calculus, only in the purely linear case. From this it becomes apparent that GoI models amount to models of combinatory algebras, but do not provide directly semantics for $\lambda$-calculus,  \ie\ deep $\beta$-reductions.

Closed affine $\lambda$-terms  are interpreted in the algebra of partial involutions ${\mathcal I}$ via a standard abstraction procedure which maps $\lambda$-terms into terms of combinatory logic and preserves  principal types.

The main result of this paper consists in showing, by elementary arguments, that the game theoretic semantics of any closed affine $\lambda$-term, not necessarily in normal form, amounts to the partial involution formally corresponding to  its principal type. This is achieved by proving that the partial involution obtained via the GoI application between the partial involutions induced by the principal type of closed terms $M,N$, corresponds to the principal types of the $\lambda$-term $MN$ obtained via type resolution.

Once the above correspondence between resolution of principal types and GoI application of the induced partial involutions has been established, properties of the type system can be deduced from properties  of partial involutions, and vice versa.
For instance, we can derive that  the use of combinators is not really a matter of choice in dealing with GoI. Namely, since full subject reduction holds for principal types only on the purely linear fragment, but not in the affine case, we derive that the model of partial involutions is a $\lambda$-algebra only on the purely linear fragment, while it is only a combinatory algebra already on the affine fragment.   This is  the reason for the extra heavy quotienting machinery which needs to be introduced in the literature on game semantics in order to achieve cartesian closed categories.

In \cite{HLS20}  we describe the tool $\Lambda$-symsym, available at \burl{https://lambda-iot.uniud.it/automata/} which
allows for computing with partial involutions and their corresponding principal types, even for a larger fragment than purely affine $\lambda$-calculus. One can use it to readily  machine check ``experimentally'' that all Curry's equations hold for purely linear combinators  ${\mathbf B}, {\mathbf C},{\mathbf I}$, thereby showing that a purely linear combinatory algebra is actually a purely linear $\lambda$-algebra.
\smallskip

\noindent {\bf Related Work.}  Geometry of Interaction for proof nets was introduced by Girard in a series of papers, (\cite{Girard1,Girard11,Girard2}), and further developed in various directions, such as that of token machines (see \eg\ \cite{DH96,M95}), context semantics (see \eg\ \cite{Gal92}), and traced monoidal categories (see \eg\ \cite{AJ94,AJ94bis,AHS02}). In particular, in \cite{DR93}, a computational view of GoI has been developed, by providing a compositional translation of the $\lambda$-calculus into a form of reversible abstract machine.

Partial involutions and their variants for building GoI models have been introduced in \cite{Abr96,Abr97}. These models can be viewed as instances of a general categorical GoI construction based on traced monoidal categories (\cite{AHS02}). Partial involutions and their variants have been used for providing GoI semantics for different  type and untyped  theories, and  various models of computation  (see \cite{AL00,AL05,AL01,Abr05,CHLS18,LC,CDGHLS19}), and as a model of reversible computation in \cite{Abr05}.

A formulation of GoI on proof nets using resolution has been introduced in \cite{Girard2}. That approach, as spelled out in \cite{Baillot01}, allows to derive the ultimate results of the present paper,  but only up to appropriate reformulations, and through a completely different route.  Those results rely on a rather complex framework, based on proof nets and $C^*$-algebras. Our approach, on the other hand, is based on the apparently hitherto unappreciated analogy between partial involutions and principal types, and it permits therefore to achieve those results using the very basic framework of combinators and type substitutions, thus providing an alternative elementary explanation of GoI. In this paper we show that GoI application, and hence the denotational semantics of affine terms, is simply resolution of principal types, and, as shown in Section \ref{twoviews}, GoI application is just a variable occurrence oriented way of carrying out unification. This is a contribution to the conceptual understanding of GoI.

Finally, we point out that in a recent paper by \cite{DL22}, principal types have been related to another GoI model, namely $\lambda$-nets. The problem has been explored for a specific type assignment system,  including intersection and bang operators on types.
The principal types for this system together with unification have been shown to correspond to $\lambda$-nets with a non-standard notion of cut-elimination.  This correspondence allows for deriving  properties of the type system, such as typability, subject reduction, and inhabitability, from properties of $\lambda$-nets, and vice-versa.
In our view, the work in \cite{DL22} and the present paper, by relating two different GoI models to principal types, pave the way to study the connections between principal types and many other GoI models.  General forms of principal types are an essential tool for exploring the fine structure of a plethora of GoI models arising in different contexts.
\smallskip

\noindent {\bf Synopsis.} In Section~\ref{sec-calcolo}, we  recall the notions of affine $\lambda_A$-calculus and affine combinatory logic, and the abstraction algorithm for encoding $\lambda_A$-calculus into combinatory logic.
In Section~\ref{tasy}, we introduce the type system for assigning principal types to affine $\lambda$-terms, and we study its properties. In Section~\ref{modin}, we introduce the model of partial involutions induced by binary types, and we study its fine structure. In Section~\ref{rgs}, we study the relationship between the GoI semantics of partial involutions and principal types, and we
 formalize the alternate characterization of unification in terms of GoI which emerges.
 Final remarks and directions for future work appear in Section~\ref{finrem}.
\smallskip

\noindent {\bf Acknowledgments.} The authors would like to express their gratitude to the referees for their patient work and suggestions.

\section{Affine $\lambda$-calculus and Combinatory Logic}\label{sec-calcolo}
In this section, we  recall the notions of affine $\lambda_A$-calculus and affine combinatory logic, and the abstraction algorithm for encoding $\lambda_A$-calculus into combinatory logic. For further details see \cite{HS86}.

\begin{definition}[Affine $\lambda$-calculus, Combinatory Logic]\hfill  \label{lambda}
\\ (i)
The language {$\mathbf{\Lambda}_{A}$} of the {\em  $\lambda_A$-calculus} is inductively defined from variables $x,y,z, \ldots$ and it is closed under  the following  formation rules:\smallskip \\
\begin{tabular}{ll}
  \reguno{M \in \mathbf{\Lambda}_A\ \ \  N \in \mathbf{\Lambda}_A }{MN \in \mathbf{\Lambda}_A} $(\mbox{app})$
& \ \ \ \ \ \ \ \ \
 \reguno{M \in \mathbf{\Lambda}_A \ \ \ o(x,M)\leq 1 }{\lambda x.M \in \mathbf{\Lambda}_A} $(\lambda)$
\end{tabular} \smallskip

\noindent where $o(x,M)$ denotes  the number of occurrences of the variable $x$ in  $M $.
 \emph{Closed terms} in $\mathbf{\Lambda}_{A}$ are denoted by $\mathbf{\Lambda}_A^{0}$.
\\
The  {\em reduction rules} of the $\lambda_{A}$-calculus are the following: \smallskip \\
\begin{tabular}{ll}
  $(\lambda x. M)N \rightarrow_{A} M[N/x]$   $(\beta )$
  &   \reguno{M\rightarrow_{A} N \ \  \lambda x. M \in \mathbf{\Lambda_{A}}}{\lambda x. M \rightarrow_{A} \lambda x.N} $(\xi)$    \medskip \\
 \reguno{M_1  \rightarrow_{A} M'_1 }{M_1 M_2 \rightarrow_{A} M'_1 M_2 }  $(\mathit{app_L})$   &     \reguno{M_2  \rightarrow_{A} M'_2 }{M_1 M_2 \rightarrow_{A} M_1 M'_2 }  $(\mathit{app_R})$ \medskip
\end{tabular}\smallskip \\
We denote by $ \rightarrow_{A}^*$ the reflexive and transitive closure of  $\rightarrow_{A} $, and by $=_A$ the conversion relation.
\\ {\em Normal forms} are  non-reducible terms.
\\ (ii)  The set of terms  $\mathbf{CL}_A$ of {\em affine combinatory logic} includes variables, combinators $ {\bf B}, {\bf C}, {\bf I}, {\bf K}$, and it is closed under application. Closed $\mathbf{CL}_A$-terms are denoted by
$ \mathbf{CL}_A^0$.
\\  The {\em reduction rules} of $\mathbf{CL}_A$ are the following:
\[   {\bf B} MNP  \rightarrow_A M(NP) \ \ \  \ \ \ \  {\bf C} MNP \rightarrow_A   (MP)N \ \ \ \ \ \ \ {\bf I} M \rightarrow_A M \ \ \  \ \  \ \  {\bf K} MN  \rightarrow_A  M \ .\]
\noindent We denote by $ \rightarrow_{A}^*$ the reflexive and transitive closure of  $\rightarrow_{A} $, and by $=_A$ the conversion relation.
\end{definition}

\noindent {\bf Notation.}  Throughout the paper, we use $\equiv$ to denote syntactic equality. Free and bound variables are defined in the standard way (see \cite{Bar84} for more details); for a given term $M$, we denote by $FV(M)$ the set of free variables of $M$.
 Given a relation $R$, we denote by $R^{op}$ the symmetric relation, \ie\ $\{ (b,a)\ |\ (a,b) \in R\}$. \medskip

It is well-known that the affine $\lambda_A$-calculus can be encoded into combinatory logic,  preserving {\em top-level $\beta$-reduction}, \ie\  $\beta$-reduction not inside $\lambda$'s:

\begin{definition}
We define  two homomorphisms w.r.t. application:
\\ (i) $(\ )_{\lambda}: \mathbf{CL}_A\rightarrow \mathbf{\Lambda}_A$, given a term $M$ of   $ \mathbf{CL}_A$, yields the term of $ \mathbf{\Lambda}_A$ obtained from $M$
by substituting, in place of each combinator, the corresponding $ \lambda_A$-term as follows
\\ \begin{tabular}{llll}
 $ ({\bf B})_{\lambda}= \lambda x y z. x(yz) $ \ \ \ \ \ \ \ \ & $   ({\bf C})_{\lambda} = \lambda x y z.(xz)y   $  \ \ \ \ \ \ \ \ &
$  ({\bf I})_{\lambda}= \lambda x. x  $ \ \ \ \ \ \ \ \ &  $   ({\bf K})_{\lambda} = \lambda x  y . x   $
 \end{tabular}
\\ (ii)  $(\ )_{CL}: \mathbf{\Lambda}_A\rightarrow \mathbf{CL}_A$, given a term  $M$ of the  $\lambda_A$-calculus, replaces each $\lambda$-abstraction
by a $\lambda^*$-abstraction.  Terms with $\lambda^*$-abstractions  amount to  $\mathbf{CL}_A$-terms obtained via the \emph{Abstraction Operation} defined below. \end{definition}

\begin{definition}[Affine Abstraction Operation] \label{op} The following operation is defined by induction on  terms of $\mathbf{CL}_A$:
\\  $\lambda^{*}x. x  = {\bf I}$ \ \ \   $\lambda^{*}x. y  = {\bf K}y$\ , \ for   $x\neq y$
\\ $\lambda^* x. MN = \begin{cases}
{\bf C} (\lambda^* x.M) N & \mbox{ if } x\in FV(M)
\\ {\bf B}M (\lambda^* x.N)  & \mbox{ if }  x\in FV(N)
\\ {\bf K}(MN) & \mbox{ otherwise.}
\end{cases}$
\end{definition}

\begin{thermo}[Affine Abstraction, \cite{HS86}] \label{absth}
For all terms $M,N \in \mathbf{CL}_A$, $(\lambda^* x.M) N =_A M[N/x]$.
\end{thermo}
\begin{proof}By straightforward induction on the definition of $\lambda^*$.
\end{proof}

Here we recall the notion of {\em affine combinatory algebra}, \ie\  model of affine combinatory logic:

\begin{definition}[Affine Combinatory Algebra, {\bf BCK}-algebra]\label{lca}
\hfill
\\ (i) An {\em  affine combinatory algebra}  $\mathcal{ A}= (A, \cdot)$ is
an applicative structure with distinguished
elements (combinators) ${\bf B}, {\bf C}, {\bf I}, {\bf K}$ satisfying the following
equations: for all $x,y,z \in A$,
\smallskip

\begin{tabular}{llll}
  $ {\bf B}xyz =_{\mathcal A} x(yz) $ \ \ \ \ \ \ \ \ &
$  {\bf C} xyz =_{\mathcal A} (xz)y   $ \ \ \ \   \ \ \ \ & $  {\bf I} x =_{\mathcal A} x  $    \ \ \ \ \ \ \ \ &  $   {\bf K}  x  y =_{\mathcal A} x  $
\end{tabular}
\\ (ii) For an affine combinatory algebra $\mathcal{ A}$, we define  $\p{\ }^{\mathcal{A}}: \mathbf{CL}^0_A\rightarrow \mathcal{A}$ as  the natural interpretation of closed terms of  $ \mathbf{CL}_A$ into $\mathcal{A}$.
\\ (iii) Closed terms of the $\lambda_A$-calculus are interpreted on ${\mathcal A}$ via the encoding into $\mathbf{CL}_A$, \ie, for any $M\in \mathbf{\Lambda}^0_A$, we define, by a small abuse of notation,
$\p{\ }^{\mathcal A}: {\mathbf \Lambda}^0_A \rightarrow {\mathcal A}$ by: $\p{M}^{\mathcal A} = \p{(M)_{CL}}^{\mathcal A}$.
\end{definition}

Notice that combinator $\mathbf{I}$ is redundant in the above definition of affine combinatory algebra, namely it can be defined in terms of the other combinators, \eg\ as $(\mathbf{C}\mathbf{K})\mathbf{C}$.

\section{The Type Assignment System for Principal Types} \label{tasy}
In this section, we introduce the type system for assigning principal types to affine $\lambda$-terms, and we study its properties.

We start by defining  the language of types and a language for denoting variable occurrences in types; the latter is necessary to permit  a fine analysis of variable occurrences, which we need for establishing the correspondence between types and partial involutions.

\begin{definition}[Types]\hfill
\\ (i) {\em Types} $T_\Sigma$ are binary trees whose leaves are variables $\alpha, \beta, \ldots\in TVar $, and nodes are denoted by $\multimap$, \ie\
\[ (T_\Sigma \ni)\ \sigma, \tau :: = \ \alpha \ | \ \beta \ | \ \ldots \ | \ \sigma \multimap \tau \ . \]
\\ (ii)  A type $\sigma$  is \emph{binary} if each variable in $\sigma$ occurs at most twice.
\\ (iii)  \emph{Occurrences} of variables in types are denoted by terms of the shape:
 $$(O_\Sigma \ni)\ u[\alpha]::= [\alpha ] \mid lu[\alpha]\mid ru[\alpha]\ ,$$
  \noindent where
  \begin{itemize}
  \item $[\alpha]$ denotes the occurrence of the variable $\alpha$ in the type $\alpha$,
  \item if $u[\alpha]$ denotes an occurrence of $\alpha$ in $\sigma_1$ ($\sigma_2$), then $lu[\alpha]$ ($ru[\alpha]$) denotes the corresponding occurrence of $\alpha$ in $\sigma_1 \multimap \sigma_2$.
  \end{itemize}
\noindent (iv)  The {\em path} of an occurrence $u[\alpha ]$ is $u$.
\end{definition}

In Proposition~\ref{occurence-vs-types} below, we clarify the relationships between binary types and partial involutions. We recall that a partial function  $f$ on a set $A$, $f:A\rightarrow A$, is a partial involution if $f=f^{op}$.
 The correspondence  between binary types and partial involutions  arises from viewing a type as a set of variable occurrences. If the type is binary, moreover,   the set of pairs of occurrences of the same variable can be read as the graph of a partial involution with domain $O_{\Sigma}$. Vice versa, from a set of variable occurrences such that no path is the initial prefix of any other path of a different occurrence, we can build the tree of a type, by tagging possible missing leaves with fresh variables. The following proposition formalizes the structural correspondence between types and sets of variable occurrences. The proof easily follows from the definitions.

\begin{proposition}\label{occurence-vs-types}
\hfill
\\ (i)
 A type  $\tau$ gives rise to a set of variable occurrences $${\mathcal O}(\tau)=\{u[\alpha]\  \mid \ u[\alpha] \mbox{ is an occurrence of the type variable } \alpha \mbox{ in } \tau\}.$$
\\ (ii)  A \emph{binary} type $\tau$ gives rise to a set of pairs of occurrences $${\mathcal R}(\tau)=\{<u[\alpha],v[\alpha] >\ \mid \  u[\alpha],v[\alpha] \mbox{ are different occurrences of } \alpha \mbox{ in } \tau\}.$$
\noindent  Then ${\mathcal R}(\tau)$ is a partial involution on $O_{\Sigma}$.
\\ (iii)  Given a set of variable occurrences $S$ where no path is the initial prefix of any other path of a different occurrence, and let $Z=\{\zeta_1,\ldots,\zeta_i,\ldots\}$ be an infinite set of  fresh variables. The following is a type:
$${\mathcal T}_Z({S})= \begin{cases} \zeta & \mbox{ if } {\mathcal S}=\emptyset \\
\alpha & \mbox{ if } {\mathcal S}=\{ [\alpha] \} \\
{\mathcal T}_Z(\{u\ \mid\  lu\in {\mathcal S} \})\multimap {\mathcal T}_Z(\{u\ \mid\ ru\in {\mathcal S}\}) & \mbox{ otherwise},
\end{cases} $$
\noindent where $Z$-variables in ${\mathcal T}_Z({S})$ are taken all different.
\\ (iv)  For all type $\sigma$, we have $${\mathcal T}_Z({\mathcal O}(\sigma))=\sigma,$$ taking equality of types up-to injective renaming of variables.
\end{proposition}

In the following definition, we introduce two  notions, which will be useful in the sequel. The first if that  of {\em $Z$-ancestor} of a given type $\sigma$: a type $\sigma'$ is a Z-ancestor of $\sigma$ if it coincides with $\sigma$ apart from variables in $Z$, which have to be suitably instantiated to obtain $\sigma$. The second notion is that of  {\em $\Theta$-free $Z$-ancestor} of  a type $\sigma$:
given a set of variables $\Theta$, the  $\Theta$-free $Z$-ancestor of $\sigma$  is the $Z$-ancestor of $\sigma$ obtained by replacing all subtypes of $\sigma$ containing only occurrences of variables in $\Theta$ with variables in $Z$.

\begin{definition} \label{ancestor} Let $\sigma$ be a type and let $\{\zeta_1,\ldots,\zeta_i,\ldots\}$ be an infinite set of  variables not appearing in $\sigma$. \\
(i) The type $\sigma'$ is a {\em $Z$-ancestor} of the type $\sigma$ if there exists a substitution $U$ such that  $U(\sigma')=\sigma$ and only variables in $Z$ are affected by $U$. \\
(ii) Let $\Theta$ be a set of variables. The {\em $\Theta$-free $Z$-ancestor} of $\sigma$, ${\mathcal F}^\Theta_Z(\sigma)$, is defined as follows: \\
${\mathcal F}^\Theta_{Z}(\sigma)= \begin{cases} \zeta & \mbox{ if } var(\sigma)\subseteq\Theta\\
\alpha &\mbox{ if } \sigma=\alpha \ \wedge\ \alpha \notin\Theta\\
{\mathcal F}^\Theta_Z(\sigma_1)\multimap {\mathcal F}^\Theta_Z(\sigma_2)& \mbox{ otherwise,}
\end{cases}$
\\ where $Z$-variables in ${\mathcal F}^\Theta_{Z}(\sigma)$ are taken all different.
\end{definition}

The proof of the following proposition easily follows from Definition \ref{ancestor} and Proposition \ref{occurence-vs-types}.

\begin{proposition}\label{extension}
Given a set of occurrences $S$, where no path is the initial prefix of any other path of a different occurrence, the type ${\mathcal T}_Z({S})$ in Proposition~\ref{occurence-vs-types}(iii) is the unique common $Z$-ancestor of all types which include the set of variable occurrences $S$.
\end{proposition}

In the following, we recall  basic definitions on  {\em type unifiers} and we present Martelli-Montanari's {\em unification algorithm}, \cite{MM82},  which  refines the original one by Robinson, and  will be used in the definition of the type  system for
assigning principal types.

\begin{definition}[Type Unifiers] Let  $\sigma$ and $\tau$ be  types.
\\ (i)
A \emph{unifier} for $\sigma$ and $\tau$ is a substitution $U$ which differs from the identity on a finite number of variables and such that $U(\sigma)= U(\tau)$.
We call {\em domain} of $U$, $\mathit{dom} (U)$, the finite set of variables on which $U$ is not the identity.
\\ (ii) Given two substitutions $U$ and $V$ whose domains are disjoint, we define the {\it union}, $U\cup V$ as usual.
\\ (iii) Given two substitutions, $U$ and $V$, the composition, $V\circ U$, is defined as usual. %  and it will often simply denoted by $VU$.
\\ (iv) Given  two substitutions, $U$ and $V$,  we  define $U\leq V$ if there exists a substitution $U'$ such that $U'\circ U = V$, \ie\ $U$ is more general than $V$.
\\ (v) Given types $\sigma,\tau$, the {\em most general unifier (m.g.u.)} of $\sigma$, $\tau$ is a unifier $\overline{U}$ of $\sigma$ and $\tau$ such that, for any unifier $U$ of $\sigma$ and $\tau$, $\overline{U}\leq U$.
\end{definition}

The following proposition follows immediately from \cite{MM82}, and can be proved by induction on the complexity of pairs of types, using the appropriate measure.
\begin{proposition}[Unification Algorithm]\label{martelli}
Let $E$ be  a set of pairs of types. The following non-deterministic algorithm  computes, if it exists, \emph{the most general unifier} (m.g.u.) of a given set of pairs of types, otherwise it yields {\tt fail}:

\noindent

\begin{tabular}{ccc}
  $MGU(\{ \langle \sigma_1 \multimap \sigma_2 ,  \tau_1 \multimap \tau_2\rangle \} \cup E )$ &$\rightarrow$& $MGU(\{ \langle \sigma_1, \tau_1 \rangle,  \langle \sigma_2 , \tau_2 \rangle\} \cup E) $\\
 $MGU( \{ \langle \alpha, \alpha\rangle  \} \cup E ) $ &$\rightarrow$& $  E $ \\
$MGU(  \{\langle \sigma_1 \multimap \sigma_2 , \alpha \rangle \} \cup E) $ &$\rightarrow$& $  MGU( \{ \langle \alpha, \sigma_1 \multimap \sigma_2\rangle  \} \cup E )$ \\
$ MGU( \{ \langle \alpha , \sigma \rangle\} \cup E ) $ &$\rightarrow$& $
                  MGU( \{ \langle \alpha , \sigma \rangle \} \cup E[\sigma/\alpha ] ),
                    \mbox{      if } \alpha \not\in Var(\sigma) \ \wedge\ \alpha \in Var(E)$ \\
$MGU(  \{ \langle \alpha , \sigma \rangle \} \cup E) $ &$\rightarrow$& $  \mbox{ {\tt fail},     if } \alpha \in Var(\sigma) \ \wedge \  \alpha \neq \sigma$
\end{tabular}
When no rules can be applied, the final set gives the m.g.u. of the initial set of pairs.
\end{proposition}

In the following definition, we extend the definition of unifiers to variable occurrences.
\begin{definition}[Occurrence Unifiers] Let $\sigma$, $\tau$ be types.
\\ (i) Two occurrences $u[\alpha]\in \sigma $ and $v[\beta]\in \tau$ are \emph{unifiable} if $u$ is a prefix of $v$, \ie\ there exists $w$ such that $uw=v$, or vice versa.
\\ (ii) If two occurrences $u[\alpha]\in \sigma $ and $v[\beta]\in \tau$ are \emph{unifiable}, their \emph{occurrence unifier} (occ-unifier) is the most general unifier of  ${\mathcal T}_Z(\{u[\alpha]\})$ and ${\mathcal T}_Z(\{v[\beta]\})$.
\end{definition}

The following notation will be used in the sequel:
\begin{notation}\hfill
\\ (i) Given unifiable occurrences $u[\alpha] \in \sigma$ and $v[\beta]\in \tau$, and their occurrence unifier $U$, by $U(u[\alpha])$ we mean either $u[\alpha]$ in case $u=vw$, or the occurrence $uw[\beta]$, in case $uw=v$, and similarly for $U(v[\beta])$.
\\ (ii) Let  $U$ be the occ-unifier of   $u[\alpha]\in \sigma $ and $v[\beta]\in \tau$, and let $u'[\alpha]$ be a different occurrence of $\alpha \in \tau$, by $U(u'[\alpha])$ we mean either $u'[\alpha]$ or $u'w[\beta]$, in case  $v=uw$.
\end{notation}

\begin{remark}\label{pons-asinorum}
Let $v[\alpha]$ and $w[\beta]$ be  two different occurrences  in the type $\sigma$, then they are not unifiable, and no occ-unifier $U$ involving $\alpha$ will ever make $U(v[\alpha])$ and $U(w[\beta])$  unifiable.
\end{remark}

\subsection{The Principal Type Assignment System and its properties}

We introduce now the type system for assigning principal types:

\begin{definition}[Principal Type Assignment System]  Let $\Vdash_A$ be the following type assignment system:\medskip \\
\reguno{}{x:\alpha\Vdash_A x:\alpha} \ \ (var)\ \ \ \
\reguno{\Gamma, x:\sigma\Vdash_A M:\tau}{\Gamma\Vdash_A \lambda x. M : \sigma\multimap \tau} \ \ (abs)\ \ \ \ \reguno{\Gamma \Vdash_A M:\sigma \ \ \ \alpha \mbox{ fresh }}{\Gamma\Vdash_A \lambda x. M : \alpha\multimap \sigma} \ \ (abs$_\emptyset$)
\bigskip

\noindent \reguno{\begin{tabular}{c} $\Gamma\Vdash_A M: \sigma \ \ \ \Delta\Vdash_A N:\tau \ \ \ (\mathit{dom}(\Gamma) \cap \mathit{dom}(\Delta))= \emptyset  \ \ \  (\mathit{TVar}( \Gamma)\cap \mathit{TVar}( \Delta))=\emptyset   $ \vspace{-0.15cm}
\\    $    (\mathit{TVar}( \sigma)\cap \mathit{TVar}( \tau))=\emptyset  \ \ \  U'= MGU(\sigma, \alpha\multimap \beta)\ \ \ U=MGU(U'(\alpha), \tau)
\ \ \ \alpha, \beta \ \mathit{fresh}$  \end{tabular} }{ U  \circ  U' (\Gamma, \Delta) \Vdash_A MN:U \circ  U'(\beta) }\ \  (app)\medskip
\\ where $MGU(\sigma,\tau)$ denotes the m.g.u.  between the types $\sigma, \tau$, which can be computed, say, via the  unification algorithm of Proposition~\ref{martelli}.
\end{definition}

The type assignment  system $\Vdash_A$ satisfies a number of  remarkable properties:

\begin{enumerate}
\item it assigns a unique type (up-to injective renaming of type variables) to each $\lambda$-term; \label{unop}
\item all judgements derivable in $\Vdash_A$ are {\em binary}, \ie\ each type variable occurs at most twice in a judgement $\Gamma\Vdash_A M:\sigma$; \label{duep}
\item judgements derivable in $\Vdash_A$ are {\em principal} w.r.t. a simple type assignment system $\vdash_A$ (which we introduce below), in the sense that all judgments derivable in $\vdash_A$ are instances of the unique judgement derivable in $\Vdash_A$; \label{tre}
\item principal types are preserved by $\beta$-conversion not within $\lambda$'s; \label{quattro}
\item different normal forms receive different principal types; \label{5}
\item the encoding of $\lambda$-terms into combinatory logic preserves principal types.\label{6}
\end{enumerate}

Most of these results, in slightly different frameworks, have appeared before in the literature, see \eg\ \cite{H89, Hi93, Ma04}, and in some cases references in these papers point to even earlier work.  In order to make the present paper self-contained and homogenous in the formalism, and to give evidence of the elementary nature of our approach, we provide  our own proofs, rather than referring to such results.

In the following proposition, we address Properties \ref{unop} and \ref{duep}. In particular, as far as Property~\ref{unop}, we show that, if a term is typable, then its principal type is unique up-to injective substitution.
Below (see Corollary~\ref{unic}) we complete the proof of Property~\ref{unop}, by showing that all $\lambda$-terms are typable with principal type.

\begin{proposition}\hfill \label{unique}
\\ (i) For any $M\in {\mathbf \Lambda}_A$, if the judgements $\Gamma \Vdash_A M: \sigma$ and $\Gamma' \Vdash_A M: \sigma'$ are derivable, then there exists an injective substitution $U$ such that $U(\Gamma)= \Gamma'$ and $U(\sigma)=  \sigma'$.
\\ (ii) Any judgement  $\Gamma \Vdash_A M: \sigma$ derivable in $\Vdash_A$  is binary.
\end{proposition}
\begin{proof} \hfill
\\ (i) Straightforward, by induction on type derivations.
\\(ii) We proceed by induction on the structure of the derivation of the type judgement. The only critical case is application. So assume that $U'(\alpha\multimap \beta)=\sigma_1\multimap \sigma_2$, where $U'= MGU(\sigma, \alpha \multimap \beta)$. If $\sigma_1\multimap \sigma_2$ and $\tau$ are binary, then we show that ${U}(\sigma_2)$ is binary, where $U=MGU(\sigma_1, \tau)$.
To this end, we prove simultaneously, by induction on the number of steps in a successful unification procedure in Definition \ref{martelli}, that a variable occurs at most in two different pairs in $E$ and, if this occurs, then it occurs uniquely in each one of them. The base case derives from the assumption that both $\sigma_1$ and $\tau$ are binary and disjoint. The induction step is pleasingly straightforward.
If the procedure terminates, the resulting substitution applied to $\sigma_2$ clearly leaves it binary, since the only variables which can be affected by the substitution must occur only once in $\sigma_2$,  being the type binary.
\end{proof}

In order to prove Property~\ref{tre} above, we introduce a simple type assignment system:

\begin{definition}[Simple Type Assignment System]  Let $\vdash_A$ be the following type assignment system\label{ptsystem}:\medskip \\
\reguno{}{x:\sigma \vdash_A x:\sigma} \ \ (var)\ \ \ \
\reguno{\Gamma, x:\sigma\vdash_A M:\tau}{\Gamma\vdash_A \lambda x. M : \sigma\multimap \tau} \ \ (abs)\ \ \ \ \reguno{\Gamma \vdash_A M:\sigma \ \ \  x \mbox{ fresh }}{\Gamma\vdash_A \lambda x. M : \tau \multimap \sigma} \ \ (abs$_\emptyset$)
\bigskip

\noindent \reguno{\Gamma\vdash_A M: \sigma \multimap \tau  \ \ \ \Delta\vdash_A N:\sigma \ \ \ (\mathit{dom}(\Gamma) \cap \mathit{dom}(\Delta))= \emptyset}{\Gamma, \Delta \vdash_A MN:\tau }\ \  (app)
\end{definition}

Proposition~\ref{relat} below  clarifies the relationships between the two type systems, and it justifies calling the types assigned in the system $\Vdash$ as principal. In order to prove Proposition~\ref{relat}, we need the following lemma, which can be easily shown by induction on derivations:

\begin{lemma}\label{lco}
If $\Gamma \vdash_A M:\sigma$, then, for all substitutions $U$,  $U(\Gamma)\vdash_A M:U(\sigma)$.
\end{lemma}

\begin{proposition}\label{relat}
Let $M \in {\bf \Lambda}_A$.
\\ (i) If $\Gamma \Vdash_A M:\sigma$, then, for all substitutions $U$, $U(\Gamma) \vdash_A M:U(\sigma)$.
\\ (ii) If $\Gamma \vdash_A M:\sigma$, then there exists a derivation $\Gamma' \Vdash_A M: \sigma'$ and a  substitution $U$ such that $U(\Gamma')= \Gamma$
and $U(\sigma') = \sigma$.
\end{proposition}
\begin{proof}
Both items can be proved by induction on derivations. Lemma~\ref{lco} above is used to prove item (i) in the case of (app)-rule.
\end{proof}

In order to prove Property~\ref{quattro}, \ie\ that principal types are preserved by $\beta$-conversion not within $\lambda$'s, we first prove that this property holds for the simple type assignment system. Then, using, Proposition~\ref{relat}, we derive the property for principal types.

\begin{definition}[Top-level $\beta$-reduction and Conversion]
Let $\rightarrow_A^T$ be the reduction relation defined by  the reduction rules $\beta$, app$_L$, app$_R$ in Definition~\ref{lambda}(i), and  omitting rule $\xi$, and let
$=_A^T$ be the corresponding conversion relation.
\end{definition}

\begin{thermo}[Top-level Subject Conversion of $\vdash_A$] \label{top-lev}
Let $M,M'\in {\mathbf \Lambda}_A$ be such that $M=_A^T M'$. Then
\[ \Gamma \vdash_A M:\tau  \ \Longrightarrow\ \exists \Gamma'.\ (\Gamma' \vdash_A M':\tau \ \wedge \    (\Gamma')_{|  FV(M)\cap FV(M')} = (\Gamma)_{|  FV(M)\cap FV(M')} ) \ . \]
\end{thermo}
\begin{proof}
Let the set of all contexts $C[\ ]$ be defined by: $C[\ ] \ :: = \ [\ ] \ | \ C[\ ] P \ | \ PC[\ ] \ | \ \lambda x. C[\ ] $, and let the set of top-level contexts  $C^T[\ ]$ be defined by omitting $\lambda$-contexts. Let $M\rightarrow_A M'$.
\\ The thesis follows from the following  facts:
\\ (i) $\Gamma \vdash_A C[M]: \tau \ \Longrightarrow \ \Gamma' \vdash_A C[M'] : \tau$, where $\Gamma' = (\Gamma)_{|FV(M')}$;
\\ (ii) $ \Gamma' \vdash_A C^T[M'] : \tau \Longrightarrow \ \exists \Gamma\supseteq \Gamma'. \ (\Gamma \vdash_A C^T[M]:\tau )$.
\\ The proofs of the two facts above proceed by a straightforward induction on contexts; in order to deal  with the base cases, we first need to prove the following result (by induction on derivations):
\\  $\Gamma_1 \vdash_A M_1: \tau \ \wedge \ \Gamma_2  \vdash_A M_2: \sigma \ \Longleftrightarrow\ (\Gamma_1, \Gamma_2)_{| FV(M_1[M_2/x])} \vdash_A M_1[M_2/x]:\tau  \ \wedge \ \Gamma_2  \vdash_A M_2: \sigma$.
\end{proof}

Then we have:

\begin{thermo}[Top-level Subject Conversion of $\Vdash_A$] \label{Top-lev}
Let $M,M'\in {\mathbf \Lambda}_A$ be such that $M=_A^T M'$. Then
\[ \Gamma \Vdash_A M:\tau  \ \Longrightarrow\ \exists \Gamma'.\ (\Gamma' \Vdash_A M':\tau \ \wedge \    (\Gamma')_{|  FV(M)\cap FV(M')} = (\Gamma)_{|  FV(M)\cap FV(M')} ) \ . \]
\end{thermo}
\begin{proof}
Let $M,M'\in {\mathbf \Lambda}_A$ be such that $M=_A^T M'$ and $\Gamma \Vdash_A M:\tau $.  By Proposition~\ref{relat}(i), $\Gamma \vdash_A M:\tau $, then  by
Theorem~\ref{top-lev} there exists $\Gamma'$ such that  $\Gamma' \vdash_A M':\tau $ and $(\Gamma)_{| FV(M)\cap FV(M')}= (\Gamma')_{| FV(M)\cap FV(M')} $. By Proposition~\ref{relat}(ii), there exist $\overline{\Gamma}, \overline{\tau}$ such that $\overline{\Gamma}\Vdash_A M':\overline{\tau} $, and  $\Gamma', \tau$  are  instances of $\overline{\Gamma},\overline{\tau}$. To obtain the thesis, we are left to show that also $(\overline{\Gamma})_{| FV(M)\cap FV(M')}$, $\overline{\tau}$ are instances of $(\Gamma')_{| FV(M)\cap FV(M')}$, $\tau$. Namely, by Proposition~\ref{relat}(i), $\overline{\Gamma} \vdash_A M':\overline{\tau}$, hence by Theorem~\ref{top-lev} there exists  $\overline{\Gamma}'$ such that  $\overline{\Gamma}' \vdash_A M:\overline{\tau}$ and $(\overline{\Gamma}' )_{| FV(M)\cap FV(M')}= (\overline{\Gamma})_{| FV(M)\cap FV(M')}$, therefore, from $\Gamma \Vdash_A M:\tau $, by Propositions~\ref{relat}(ii),  it follows that $\overline{\Gamma}'$, $\overline{\tau}$ are instances of $\Gamma$, $\tau$. But then, since $(\overline{\Gamma}' )_{| FV(M)\cap FV(M')}= (\overline{\Gamma})_{| FV(M)\cap FV(M')}$, we also have that $(\overline{\Gamma}')_{| FV(M)\cap FV(M')} $, $\overline{\tau}$ are instances of $(\Gamma)_{| FV(M)\cap FV(M')}$, $\tau$. Finally, since $({\Gamma} )_{| FV(M)\cap FV(M')}= ({\Gamma'})_{| FV(M)\cap FV(M')}$, we have that  $(\overline{\Gamma})_{| FV(M)\cap FV(M')}$, $\overline{\tau}$ are instances of $(\Gamma')_{| FV(M)\cap FV(M')}$, $\tau$.
\end{proof}

\begin{corollary}\label{convp}
Let $M,M'\in {\mathbf \Lambda}_A^0$ be such that $M=_A^T M'$. Then
\[   \Vdash_A M:\tau  \ \Longleftrightarrow\  \Vdash_A M':\tau   \ .\]
\end{corollary}

The following counterexample shows that  subject conversion  fails in general, when also reduction under $\lambda$'s is considered.

\begin{counterexample}
Let us consider $\lambda x y z. (\lambda w. x)(yz)$ and its $\beta$-reduct $ \lambda x y z.x$. We have:
\\ $\Vdash_A \lambda x y z.x: \alpha_1 \rightarrow \alpha_2 \rightarrow \alpha_3 \rightarrow \alpha_1$, but we cannot derive $  \Vdash_A \lambda x y z. (\lambda w. x)(yz)  : \alpha_1 \rightarrow \alpha_2 \rightarrow \alpha_3 \rightarrow \alpha_1$. We can derive only $ \Vdash_A \lambda x y z. (\lambda w. x)(yz): \alpha_1 \rightarrow (\alpha_2 \rightarrow \alpha_3) \rightarrow \alpha_2 \rightarrow \alpha_1$, which is an instance of the former, because the variables which are erased, are erased after having been applied, and the principal type keeps track of this.
\end{counterexample}

This phenomenon arises in the affine case, but in the purely linear case subject reduction can be proved to hold in full form, \cite{CDGHLS19}.
As a consequence, principal types induce an affine combinatory algebra which is a linear $\lambda$-algebra, in the sense of \cite{Bar84}, only on the purely linear fragment.

Namely, denoting by $\Lambda_L$ the set of linear $\lambda$-terms, \ie\ terms where each variable appears exactly once, and by $=_L$, $\vdash_L$, $\Vdash_L$ the corresponding conversion and type systems, we have:

\begin{proposition}[Subject Conversion for Linear $\lambda$-calculus]
Let $M,M'\in \Lambda_L$ be such that $M=_L M'$. Then
\[ \Gamma  \Vdash_L M:\tau  \ \Longleftrightarrow\ \Gamma \Vdash_L M':\tau   \ .\]
\end{proposition}

 \begin{proof}
First we show that full subject conversion holds for $\vdash_L$.  This follows the same pattern as in the proof of Theorem~\ref{top-lev}, where by linearity we can now safely consider the full class of contexts.

Then, we show subject conversion for $\Vdash_L$. Assume that $M=_LM'$ and $\Gamma \Vdash_L M:\tau$, by Proposition~\ref{relat}(i), which holds also for the purely linear case, $\Gamma\vdash_L M:\tau$, and by subject conversion of $\vdash_L$ we have $\Gamma \vdash_L M':\tau$. By Proposition~\ref{relat}(ii), there exist $\overline{\Gamma}$,
 $\overline{\tau}$ such that $\overline{\Gamma} \Vdash_L M':\overline{\tau}$ and $\Gamma$, $\tau$ are instances of
 $\overline{\Gamma}$, $\overline{\tau}$. Then, by Proposition~\ref{relat}(i),  $\overline{\Gamma}\vdash_L M':\overline{\tau}$, and by subject conversion of $\vdash_L$, we have $\overline{\Gamma} \vdash_L M:\overline{\tau}$. By Proposition~\ref{relat}(ii), there exist $\overline{\Gamma}'$, $\overline{\tau}'$ such that $\overline{\Gamma}' \Vdash_L M:\overline{\tau}'$ and  $\overline{\Gamma}$, $\overline{\tau}$   are instances of $\overline{\Gamma}'$, $\overline{\tau}'$. From $\Gamma\Vdash_L M:\tau$ and  $\overline{\Gamma}'\Vdash_L M:  \overline{\tau}'$, by Proposition~\ref{relat}(i), $\Gamma$, $\tau$ coincide with  $\overline{\Gamma}'$, $\overline{\tau}'$ up-to injective substitution. Hence, since both  $\Gamma$, $\tau$ are instances of $\overline{\Gamma}$, $\overline{\tau}$, and  $\overline{\Gamma}$, $\overline{\tau}$ are instances of $\overline{\Gamma}'$, $\overline{\tau}'$, we have that
 also
 $\Gamma$, $\tau$ and $\overline{\Gamma}$, $\overline{\tau}$ coincide up-to injective substitution. Therefore, from  $\overline{\Gamma} \Vdash_L M':\overline{\tau}$, we finally  have $\Gamma \Vdash_L M': \tau$.
 \end{proof}

We are now in the position of proving that all affine $\lambda$-terms receive a unique type in $\Vdash$ (up-to injective substitution on types). We first prove that  all terms receive a type in the simple type assignment system $\vdash_A$. Then, by Proposition~\ref{relat}(ii), all terms receive a principal type in $\Vdash$, which is unique by Proposition~\ref{unique}(i).
This completes the proof of Property~\ref{unop}. We start by proving the following lemma:

\begin{lemma} \label{nft}
\hfill
\\ (i) All normal forms are typable in $\vdash_A$.
\\ (ii) Let $M,M'\in \Lambda_A$ be such that $M'\rightarrow_A M$. Then
\\ $ \Gamma \vdash_A M:\tau \ \Longrightarrow\ \exists \Gamma',\tau', U.\  (\Gamma' \vdash_A M':\tau' \ \wedge\ U(\Gamma)=  (\Gamma')_{| FV(M)} \ \wedge\ U(\tau)= \tau')$.
\end{lemma}
\begin{proof}
\hfill
\\ (i) Straightforward, by induction on the structure of normal forms.
\\ (ii) Let $C[(\lambda x.M)N] \rightarrow_A C[M[N/x]]$, and $\Gamma \vdash_A C[M[N/x]]:\tau$.  We proceed  by induction on the structure of the context $C[\ ]$.  If $C[\ ]
\equiv [\ ]$, then the thesis follows from Theorem~\ref{top-lev}. If $C[\ ]\equiv C_1[\ ] P$, then from $\Gamma \vdash_A C_1[M[N/x]]P:\tau$ it follows that there exist $\Gamma_1, \Gamma_2$, $\sigma$ such that $\Gamma \equiv \Gamma_1, \Gamma_2$, $\Gamma_1 \vdash_A C_1[M[N/x]]:\sigma\multimap \tau$ and  $\Gamma_2 \vdash_A P:\sigma$. By induction hypothesis, there exist $\Gamma'_1$, $\sigma'$, $\tau'$, $U$ such that $\Gamma'_1 \vdash_A C_1[(\lambda x. M)N]:\sigma' \multimap \tau'$, $U(\Gamma_1) = (\Gamma'_1)_{| FV(C_1[M[N/x]])}$, $U(\sigma \multimap \tau) = \sigma'\multimap  \tau'$. By Lemma~\ref{lco}, $U(\Gamma_2) \vdash_A P: U(\sigma)$. Hence $U(\Gamma_1, \Gamma_2) \vdash_A  C_1[(\lambda x. M)N]P: U(\tau)$. If $C[\ ] \equiv P C_1[\ ]$, then we proceed in a  way similar to the case above. If $C[\ ] \equiv \lambda y. C_1[\ ]$, then there are various cases, depending on where the variable $y$ appears free: either in  both $C_1[(\lambda x.M)N]$  and $C_1[M[N/x]]$, or only in the first term, or nowhere. We discuss only the case where $y$ appears free in both terms, the other cases can be dealt with similarly. From $\Gamma \vdash_A \lambda y. C_1[M[N/x]]:\tau_1 \multimap \tau_2$, we have that $\Gamma, y:\tau_1  \vdash_A  C_1[M[N/x]]: \tau_2$. By induction hypothesis, there exist $\Gamma'$, $\tau'_1$, $\tau'_2$, $U$ such that  $\Gamma' , y: \tau'_1 \vdash_A  C_1[(\lambda x. M)N]:\tau'_2 $, and $U(\tau_i)= \tau'_i$, $U(\Gamma) = \Gamma'_{| FV(C_1[M[N/x]])}$. Hence we get
$\Gamma' \vdash_A \lambda y. C_1[(\lambda x. M)N]:\tau'_1 \multimap \tau'_2 $.
\end{proof}

\begin{thermo}\label{typ}
For all $M\in \Lambda_A$ there exists a judgement $\Gamma \vdash_A M: \tau$.
\end{thermo}
\begin{proof}
By induction on the number of reduction steps to normal form, using Lemma~\ref{nft}.
\end{proof}

By Theorem~\ref{typ}, Proposition~\ref{relat}, and Proposition~\ref{unique}, we finally have:

\begin{corollary}[Uniqueness] \label{unic}
Let $M\in \Lambda_A$. Then there exists a unique judgement up-to injective substitution on types, $\Gamma \Vdash_A M:\tau$, which is derivable in $\Vdash_A$.
\end{corollary}

Property \ref{5} amounts to the following proposition:
 \begin{proposition}[Principal Types determine Normal Forms]\label{p5}
 Let  $M,N\in {\mathbf \Lambda}_A$ be normal forms such that $\Gamma\Vdash M:\sigma$ and $\Gamma\Vdash N:\sigma$, then $M=_A N$.
\end{proposition}
\begin{proof} Consider the shortest pair of closed derivable judgements in $\Vdash_A$ with the same type but different terms in normal form. The rightmost  type variable in the type $\sigma$ must  occur also in the type of the head variable of the term, because all normal forms have a head variable. Since the judgements are binary, by Proposition \ref{unique}, that type variable  can occur only in the type of that term variable, so the head variable is uniquely determined and it must coincide, together with its type, in both judgements. Hence the difference between the two $\beta$-normal forms must be in the arguments of the the head variable, where the head variable does not occur, because the term is affine. These are shorter, hence we have a contradiction.
\end{proof}

In order to address Property \ref{6}, we need a proposition, whose proof is immediate:
\begin{proposition}  \label{ptycomb}\hfill \\
$\Vdash_A ({\mathbf I})_{\lambda}: \alpha \multimap \alpha$\\
$\Vdash_A ({\mathbf K})_{\lambda}: \alpha  \multimap  \beta  \multimap  \alpha$\\
$\Vdash_A ({\mathbf B})_{\lambda}: (\alpha  \multimap  \beta) \multimap  (\gamma  \multimap  \alpha) \multimap  \gamma  \multimap  \beta$\\
$\Vdash_A ({\mathbf C})_{\lambda}:( \alpha  \multimap  \beta  \multimap  \gamma) \multimap  \beta \multimap  \alpha  \multimap  \gamma$\\
\end{proposition}
The above types are the well-known principal types for the basic combinators and they will be taken as the types of the combinators in the following proposition, which addresses Property \ref{6}.

\begin{proposition}\label{p6}
Let $M\in {\mathbf \Lambda}_A$. Then  $\Gamma \Vdash_A M:\sigma \ \Longleftrightarrow\ \Gamma \Vdash_A ((M)_{CL})_{\lambda}:\sigma $.
\end{proposition}

\begin{proof}\hfill\\
($\Rightarrow$) By induction on $M$.
\\ If $M\equiv x$, then the thesis is immediate.
\\ If $M\equiv M_1 M_2$, then  the thesis follows by applying the  induction hypothesis, since $((M_1 M_2)_{CL})_{\lambda}
= ((M_1)_{CL})_{\lambda} ((M_2)_{CL})_{\lambda}$.
\\  If $M\equiv \lambda x.M'$, then  $((M)_{CL})_{\lambda}= (\lambda^* x. M')_{\lambda}$ and
there are various cases, according to the shape of $M'$. If $M'\equiv x$, then $((M)_{CL})_{\lambda}=( \mathbf{I})_{\lambda}= M$, and the thesis is immediate. If $M'\equiv y$, $y\not \equiv x$, then $((M)_{CL})_{\lambda}=( \mathbf{K}y)_{\lambda}$, and the thesis is immediate.
If $M'=M'_1M'_2$ and $x\in FV(M'_2)$, then from $\Gamma \Vdash_A \lambda x. M'_1 M'_2:\tau \multimap \sigma$, we have  $\Gamma, x:\tau \Vdash_A  M'_1 M'_2:\sigma$.
Hence there exist $\sigma_1, \sigma_2, \sigma_3$ such that $\Gamma \Vdash_A M'_1:\sigma_1\rightarrow \sigma_2$, $\Gamma, x:\tau_1 \Vdash_A M'_2:\sigma_3$, and $\sigma$ is the result of the resolution between $\sigma_1\multimap \sigma_2$ and $\sigma_3$, and $\tau$ is the result of the application of the resolvent substitution to $\tau_1$. By applying the induction hypothesis, we obtain that $\Gamma \Vdash_A ((M'_1)_{CL})_{\lambda} :\sigma_1\rightarrow \sigma_2$ and
  $\Gamma  \Vdash_A (\lambda^*x.M'_2)_{\lambda}:\tau_1\multimap \sigma_3$.  It is  straightforward, using the uniqueness of principal types, to check that applying $(\mathbf B)_{\lambda}$ first to $((M'_1)_{CL})_{\lambda}$ and then to $(\lambda^*x.M'_2)_{\lambda}$ yields the result. The remaining cases are dealt with similarly.\\
($\Leftarrow$)  Let  $\Gamma \Vdash_A ((M)_{CL})_{\lambda}:\sigma$. Then by Corollary~\ref{unic} there exist $\Gamma', \sigma'$ such that $\Gamma' \Vdash_A M:\sigma'$. Hence, by the implication $(\Rightarrow)$ above, $\Gamma' \Vdash_A ((M)_{CL})_{\lambda}:\sigma'$, therefore,  by uniqueness of principal types up-to injective substitution, we  finally have $\Gamma \Vdash_A M:\sigma$.
\end{proof}

\section{The Model of Partial Involutions}\label{modin}

In the following definition, we introduce the model of partial involutions induced by binary types.
As already observed (see Proposition~\ref{occurence-vs-types}(ii) of Section~\ref{tasy}), each binary type $\tau$  induces a partial involution ${\mathcal R}(\tau)$ on the language $O_{\Sigma}$ of type variable occurrences.
On these partial involutions, we define a notion of linear application in the  GoI-style, see \eg\ \cite{AL05,Abr05}. However, proving that partial involutions induced by binary types are closed under application requires a number of results, which build up to Corollary~\ref{clos} below. Once established this result, the applicative  structure of partial involutions is easily shown to be an affine combinatory algebra, and hence to provide a GoI semantics for $\mathbf{CL}_A$ and $\lambda_A$.

\begin{definition}[The Model of Partial Involutions $\mathcal I$]\label{cdot}\hfill
\\ (i) $ {\mathcal I}$ is the set of {\it partial involutions}  induced by binary types, \ie\ ${\mathcal I} = \{ {\mathcal R}(\tau) \ | \ \tau\in T_\Sigma \ \wedge\  \tau \mbox{ binary}\}$.
\\ (ii) Given binary types $\sigma,\tau \in T_\Sigma$, we define ${\mathcal R}(\tau)\, \hat{;}\, {\mathcal R}(\sigma)$ as  ``unification and postfix composition'', namely
\\ $ {\mathcal R}(\tau)\, \hat{;}\, {\mathcal R}(\sigma) = \{\langle U(u[\alpha]),U(v'[\beta]) \rangle \ \mid \ \langle u[\alpha],u'[\alpha]\rangle \in {\mathcal R}(\tau),\ \langle v[\beta],v'[\beta]\rangle \in {\mathcal R}(\sigma),\ $
\\ \hspace*{9cm} $U \mbox{ occ-unifier of } u'[\alpha] \mbox{ and } v[\beta]\}$\\
(iii) The notion of {\em linear application} is defined, for $f,g\in {\mathcal I}$, by  $$f\cdot g =\ f_{rr}\cup (f_{rl}\, \hat{;}\, g\, \hat{;}\, (f_{ll}\, \hat{;}\, g)^*\, \hat{;}\, f_{lr})\ ,$$ where $f_{ij} \ =\ \{ \langle u,v\rangle\ |\ \langle i(u),j(v)\rangle \in f\}$, for $i,j \in \{r,l\}$ (see Fig.~\ref{fig:lappdiag}). We take variables in different pairs of $f \cdot g$ to be disjoint.\\
(iv) We define:  ${\mathcal O}(f\cdot g)=\{u\mid\exists v.\ \langle u,v\rangle\in f\cdot g\}.$
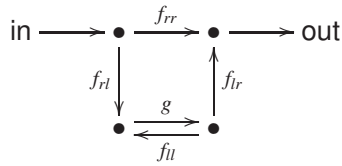
\begin{figure}[!h]
\[\xymatrix{
\mathsf{in}\ar[r] & \bullet\ar[r]^{f_{rr}}\ar[d]_{f_{rl}} & \bullet\ar[r] & \mathsf{out}\\
                  & \bullet\ar@<.5ex>[r]^{g} & \bullet\ar@<.5ex>[l]^{f_{ll}}\ar[u]_{f_{lr}} &
}\]
\caption{Flow of control in executing $f\cdot g$.}\label{fig:lappdiag}
\end{figure}
\end{definition} %\vspace*{-0.9cm}

Most  of this section will be  devoted  to proving closure of $\mathcal I$ under linear application. But to provide intuition, before we address this issue, we introduce
 the interpretation of $\mathbf{CL}_A$ on the set  of partial involutions $\mathcal I$, and, via the abstraction procedure, the interpretation of the $\lambda_A$-calculus:

\begin{definition}[GoI Semantics]\hfill\label{GoIsem}
\\ (i) The \emph{GoI semantics} of closed combinatory terms, $\pI{\ }: \mathbf{CL}_A^0\rightarrow {\mathcal I}$, is defined by induction on terms $M\in \mathbf{CL}_A^0$ as follows:

{\begin{tabular}{lll}
$\pI{\mathbf B}$ & = &   $\{ r^3 \alpha \leftrightarrow lr \alpha \ , \  l^2 \alpha \leftrightarrow rlr \alpha \ , \   rl^2 \alpha \leftrightarrow r^2 l \alpha \}$ \\
    $\pI{\mathbf I}$ & = &  $\{ l\alpha \leftrightarrow r\alpha\} $\\
 $\pI{\mathbf C}$ & = &   $\{ l^2\alpha \leftrightarrow r^2 l\alpha\ , \  lrl \alpha\leftrightarrow rl\alpha \ , \  lr^2 \alpha\leftrightarrow r^3 \alpha\} $\\
 $\pI{\mathbf K}$ & = &  $\{l\alpha \leftrightarrow r^2\alpha\}$ \\
 $\pI{MN} $ & = & $\pI{M} \cdot \pI{N}$
\end{tabular}} \label{gs} %
%\caption{Game Semantics for ${{CL}}_{A}$.}
\label{tabsem}
\\ where $u\alpha  \leftrightarrow v\alpha$ is an abbreviation for  the pairs $\langle u[\alpha], v[\alpha] \rangle$, $\langle v[\alpha], u[\alpha] \rangle$.
\\
 (ii) The GoI semantics of closed $\lambda_A$-terms, $\pI{\ }: {\mathbf \Lambda}_A^0 \rightarrow {\mathcal I}$, is defined, for any $M\in {\mathbf \Lambda}_A^0$, by: $$\pI{M} = \pI{(M)_{CL}}\ .$$
\end{definition}

Please appreciate that the semantics of the combinators given above corresponds precisely to the occurrences of the type variables in the corresponding principal types in Proposition~\ref{ptycomb}.\\
In order to clarify the working of the above semantics we  provide some examples.

\begin{example}$\pI{\mathbf{I}}\cdot \pI{\mathbf I}=\pI{\mathbf I}$\\
By  Definition \ref{GoIsem} we have:\\
\begin{tabular}{lll}
$\p{\mathbf{I}}^{\mathcal I}_{rl}= \p{\mathbf{I}}^{\mathcal I}_{lr}$&=& $\alpha\leftrightarrow \alpha$\\
$\pI{\mathbf{I}}\cdot \pI{\mathbf I}$&=& $\boldsymbol{\it r\alpha} \Leftrightarrow r\alpha \Leftrightarrow l\alpha \Leftrightarrow \boldsymbol{\it l\alpha}$
\end{tabular}\\
where we denote by $\Leftrightarrow$ the sequence of the effects of the relational compositions after the substitutions have been put into place. It is symmetric. The resulting pair appears in bold characters. $\qed$
\end{example}

\begin{example} \pI{\mathbf{CIII}}=\pI{\mathbf I}\\
By  Definition \ref{GoIsem} we have:\\
\begin{tabular}{lll}
$\p{\mathbf{C}}_{rl}^{\mathcal I} =( \p{\mathbf{C}}_{lr}^{\mathcal I})^{-1}$ & = &   $\{ rl\alpha \leftrightarrow l\alpha\ , \  r^2 \alpha\rightarrow r^2 \alpha\}$ \\
$\pI{\mathbf{C}}\cdot \pI{\mathbf I}$ & = & $\{\boldsymbol{\it r^2\alpha} \Leftrightarrow r^2\alpha \Leftrightarrow lr\alpha \Leftrightarrow\boldsymbol{\it rlr\alpha}, \boldsymbol{\it l\alpha} \Leftrightarrow rl\alpha \Leftrightarrow l^2\alpha \Leftrightarrow\boldsymbol{\it rl^2\alpha}\}$\\
&  principal type & $\alpha\multimap (\alpha \multimap \beta)\multimap \beta$\\
$\p{\mathbf{CI}}_{rl}^{\mathcal I} =( \p{\mathbf{CI}}_{lr}^{\mathcal I})^{-1}$ & = &   $\{ l^2\alpha\rightarrow\alpha \}$ \\
$\p{\mathbf{CI}}_{rr}^{\mathcal I} $ & = &   $\{ r\alpha\leftrightarrow lr\alpha \}$ \\
$\pI{\mathbf{CI}}\cdot \pI{\mathbf I}$ & = & $\p{\mathbf{CI}}_{rr}^{\mathcal I}  \cup \{\boldsymbol{\it l^2r\alpha} \Leftrightarrow r\alpha \Leftrightarrow l\alpha \Leftrightarrow\boldsymbol{\it l^3\alpha}\}$\\
$\p{\mathbf{CII}}^{\mathcal I}$ & = &   $\{  r\alpha \leftrightarrow  l r\alpha, l^3\alpha \leftrightarrow l^2r\alpha \}$\\
&  principal type & $((\alpha\multimap\alpha) \multimap \beta)\multimap \beta$\\
$\p{\mathbf{CII}}_{rl}^{\mathcal I} =( \p{\mathbf{CII}}_{lr}^{\mathcal I})^{-1}$ & = &   $\{ \alpha \rightarrow r \alpha \}$ \\
$\p{\mathbf{CII}}_{ll}^{\mathcal I}$ & = &   $\{  l^2 \alpha \leftrightarrow lr \alpha \}$\\
$\pI{\mathbf{CII}}\cdot \pI{\mathbf I}$ & = & $\{\boldsymbol{\it r\alpha} \Leftrightarrow r^2\alpha \Leftrightarrow lr\alpha \Leftrightarrow l^2\alpha \Leftrightarrow rl\alpha \Leftrightarrow\boldsymbol{\it l\alpha}\}$\\
&  principal type & $\alpha\multimap \alpha$\\
\end{tabular}\\
notice that in general the relations $\p{\ }_{rl }^{\mathcal I}$ are not symmetric. \qed
\end{example}

The proof that the set of partial involutions is closed under application requires a fine analysis and a number of technical results on linear application.

In the following proposition, we spell out application in terms of {\em trajectories}: namely,  we have $\langle u, v\rangle\in f\cdot g$ if and only if there exists a {\em trajectory}, \ie\ a suitable sequence of pairs of variable occurrences,
$\langle u_1[\alpha_1], u'_1[\alpha_1]\rangle, \ldots , \langle u_{n+1}[\alpha_{n+1}], u'_{n+1}[\alpha_{n+1}]\rangle$, together with occurrence unifiers for the pairs $\langle u'_i[\alpha_i], u_{i+1}[\alpha_{i+1}]\rangle$, for all $i$. More precisely:

\begin{proposition}\label{multiple} Let $f,g\in {\mathcal I}$. Then  $\langle u,v \rangle \in f\cdot g$ if and only if there exists a sequence, $\langle u_1[\alpha_1],u'_1[\alpha_1]\rangle,\ldots, \langle u_{n+1}[\alpha_{n+1}],u'_{n+1}[\alpha_{n+1}]\rangle$, $n$ even, such that:
\begin{itemize}
\item either $n=0$ and $\langle u_1[\alpha_1],u'_1[\alpha_1] \rangle \in f_{rr}$ or $n>0$, $\langle u_1[\alpha_1],u'_1[\alpha_1] \rangle \in f_{rl}$, $\langle u_{n+1}[\alpha_{n+1}],u'_{n+1}[\alpha_{n+1}]\rangle\in f_{lr}$,  $\langle u_i[\alpha_i],u'_i[\alpha_i] \rangle\in g$,  for $i<n$, $i$ even,  and $\langle u_i[\alpha_i],u'_i[\alpha_i] \rangle \in f_{ll}$, for $1< i < n+1$, $i$ odd;
\item the set of types $\Pi =  \{\langle {\mathcal T}_Z(u'_i[\alpha_i]),{\mathcal T}_Z (u_{i+1}[\alpha_{i+1}])\rangle  \mid 1\leq i\leq n\}$ (where  $Z$-variables used in different types are different) is unifiable with m.g.u. $U$, and $u=U(u_1[\alpha])$,  $v=U(u_{n+1}[\alpha_{n+1}])$.
\end{itemize}
The sequence $\langle u_1[\alpha_1],u'_1[\alpha_1]\rangle,\ldots, \langle u_{n+1}[\alpha_{n+1}],u'_{n+1}[\alpha_{n+1}]\rangle$ is called a	\emph{trajectory} and $\langle u,v \rangle$ its \emph{output}.
\end{proposition}
\begin{proof}
This is just a rephrasing of Definition \ref{cdot} from the perspective of ancestor types, using Proposition \ref{martelli}.
\end{proof}

In the following lemma, we study the shape of the pairs $\langle u[\alpha], v[\beta]\rangle$ belonging to ${\mathcal R}(\sigma_1\multimap \sigma_2)\cdot {\mathcal R}(\tau)$, for $\sigma_1 \multimap \sigma_2, \tau$ binary types.

\begin{lemma}\label{coherence}
Let $\sigma_1\multimap \sigma_2$ and $\tau$ be binary types and let $\langle u[\alpha], v[\beta]\rangle\in {\mathcal R}(\sigma_1\multimap \sigma_2)\cdot {\mathcal R}(\tau)$ via the trajectory $\pi$.
Then we have:
\\ (i)  $\alpha =\beta$.
\\ (ii) $u[\alpha]$ and $v[\alpha]$ are not unifiable.
\\ (iii) if $\langle u'[\alpha'], v'[\alpha']\rangle \in {\mathcal R}(\sigma_1\multimap \sigma_2)\cdot {\mathcal R}(\tau)$ is  an output pair of a trajectory, different from $\langle u[\alpha], v[\alpha]\rangle$, then $u[\alpha]$ is not unifiable with $u'[\alpha']$.
\\ (iv) $\langle  v[\alpha], u[\alpha]\rangle$ is also an output pair of a trajectory.
\end{lemma}
\begin{proof}\hfill
\\ (i) Immediate from Proposition~\ref{multiple}.
\\ (ii)   We prove the result by contradiction. If $u[\alpha]$ and $v[\alpha]$ are unifiable, then, by Remark \ref{pons-asinorum}, the components in $\sigma_2$ of the trajectory $\pi$ must coincide. But then the sequence of  pairs in the trajectory $\pi$, must be symmetric. But the sequence is odd, hence there exists a pair which is itself the identity. This is a contradiction, since by definition the $\mathcal R$'s are irreflexive.
\\ (iii) If $u[\alpha]$ and $u'[\alpha']$ are unifiable, then they must  arise from the same initial occurrence $u_1[\alpha_1]$ of $\sigma_2$ in the trajectories $\pi$ and $\pi'$. Assume that the two trajectories coincide up to stage $k$ and that they differ at stage $k+1$, because the next pairs are  $\langle u_{k+1}[\alpha_{k+1}], u'_{k+1}[\alpha_{k+1}]\rangle$ in $\pi$ and $\langle v_{k+1} [\beta_{k+1}], v'_{k+1}[\beta_{k+1}]\rangle$ in $\pi'$, with $u_{k+1}[\alpha_{k+1}]$ and $v_{k+1}[\beta_{k+1}]$ non-unifiable. But this is impossible since the pairs in the $\mathcal R$'s  arise from binary types and hence there cannot appear two pairs with the same first pair and different second pair by  Remark \ref{pons-asinorum}.
 \\ (iv)  This follows from Proposition~\ref{multiple}, since ${\mathcal R}(\sigma_1\multimap \sigma_2)$ and ${\mathcal R}(\tau)$ are partial involutions, and the set $\Pi$ of pairs of types generated by the trajectory which yields $\langle u[\alpha], v[\beta]\rangle$ is symmetric to the one yielding $\langle v[\beta], u[\alpha]\rangle$.
\end{proof}

Now we are in the position of proving that the set of pairs of variable occurrences arising from application, $\cdot$,  of partial involutions induces a binary type:

\begin{proposition}\label{make-type}
Let  $\sigma_1\multimap \sigma_2$ and $\tau$ be binary types, then ${\mathcal T}_ { Z}({\mathcal O}({\mathcal R}(\sigma_1\multimap \sigma_2)\cdot {\mathcal R}(\tau)))$ is a well-defined binary type.
\end{proposition}
\begin{proof} By Lemma \ref{coherence}(ii) and (iii), pairs of different occurrences in  $\mathcal O({\mathcal R}(\sigma_1\multimap \sigma_2)\cdot {\mathcal R}(\tau))$ are not unifiable and therefore, by Proposition~\ref{occurence-vs-types}(iii), they coherently define a type, where missing leaves are tagged by different fresh variables. ${\mathcal T}_Z({\mathcal O({\mathcal R}(\sigma_1\multimap \sigma_2)\cdot {\mathcal R}(\tau))})$ is binary by definition (variables in different pairs of the application are disjoint).
\end{proof} %let ${\mathcal O}={\mathcal R}(\sigma_1\multimap \sigma_2)\cdot {\mathcal R}(\tau)$

Finally, from the above proposition we immediately  have:

\begin{theorem}[Closure under application of $\mathcal I$]\label{clos}\hfill\\
For all binary types $\sigma_1 \multimap \sigma_2$ and $\tau$,
${\mathcal R}(\sigma_1 \multimap \sigma_2) \cdot {\mathcal R}(\tau) = {\mathcal R}({\mathcal T}_Z({\mathcal O}({\mathcal R}(\sigma_1\multimap \sigma_2)\cdot {\mathcal R}(\tau)))).$
\end{theorem}

\section{Relating GoI Semantics and Principal Types}\label{rgs}
This section is devoted to showing that the GoI semantics of a closed $\mathbf{CL}_A$-term coincides with the partial involution induced by its principal type. By Proposition~\ref{p6}, we have then that the GoI semantics of a closed $\lambda_A$-term corresponds to the partial involution induced by its principal type.

The above follows once we prove that  the partial involution obtained via the GoI application between partial involutions, induced by the principal types of the closed $\lambda_A$-terms $M,N$, corresponds to the principal type of the term $MN$. We already know, from the previous section, that ${\mathcal I}$ is closed under application, \ie\ that the result of the application $\mathcal{R}(\sigma_1 \multimap \sigma_2)\cdot \mathcal{R}(\tau)$ is a partial involution corresponding to a binary type. Here we are left to show that, if $\sigma_1$ and $\tau$ are unifiable via m.g.u. $\overline{U}$, then  $ {\mathcal R}(\sigma_1 \multimap \sigma_2) \cdot {\mathcal R}(\tau) = {\mathcal R}(\overline{U}(\sigma_2))$.
Intuitively, this is achieved by proving that, if $\sigma_1$ and $\tau$ are unifiable as types, by the m.g.u. $\overline{U}$, then the overall effect of the occ-unifiers arising from trajectories determined by the GoI application corresponds to that of $\overline{U}$.

\begin{lemma} \label{oplus} Let $U_1,U_2$ be substitutions such that, for all $i=1,2$, $U_i\leq \overline{U}$,  then the following substitution is well defined:
$$U_1 \oplus U_2 =MGU(\{\langle U_1(\beta ),U_{2}(\beta )\rangle \mid  \beta \in TVar\})\ .$$
\noindent Moreover, $U_1 \oplus U_2 \leq \overline{U}$.
\end{lemma}
\begin{proof}
Straightforward.
\end{proof}

Notice that  associativity of $\oplus$ follows from the non-deterministic nature of the MGU-algorithm.

\begin{proposition}  \label{cuore} Let $\sigma_1\multimap \sigma_2$ and $\tau$ be binary types,  let
$\Theta$ be the set of type-variables in $TVar(\sigma_1 \multimap \sigma_2,\tau)$ which are not involved in any trajectory of  ${\mathcal R}(\sigma_1\multimap \sigma_2)\cdot {\mathcal R}(\tau)$,
and let $U_{\pi_1}, \ldots , U_{\pi_n} $ be the unifiers arising from all trajectories $\pi_1, \ldots , \pi_n$ of ${\mathcal R}(\sigma_1\multimap \sigma_2)\cdot {\mathcal R}(\tau)$. If
$\sigma_1$ and $\tau$ are unifiable with m.g.u. $\overline{U}$, then:\\
(i) $U_{\pi_i}\leq \overline{U}$ for all $i$;\\
(ii) $\bigoplus_i U_{\pi_i} =\overline{U}_{\restriction{(TVar\setminus \Theta)}}$;\\
 (iii)  ${\mathcal T}_Z({\mathcal O}({\mathcal R}(\sigma_1\multimap \sigma_2)\cdot {\mathcal R}(\tau)))=\overline{U}(\sigma_2)$.
\end{proposition}
\begin{proof} \hfill
\\ (i) This follows from Proposition~\ref{multiple}.
\\ (ii) By Lemma~\ref{oplus}, $\bigoplus_i U_{\pi_i} \leq \overline{U}$.  By Proposition~\ref{multiple}, each $U_{\pi_i}$ is the m.g.u. of the set of pairs of types arising from the occurrences in the trajectory $\pi_i$. Hence, since $\sigma_1 $ and $\tau$ are unifiable, $\bigoplus_i U_{\pi_i}$ is the m.g.u. of the types $\mathcal{F}^{\Theta}_Z (\sigma_1) $ and $\mathcal{F}^{\Theta}_Z (\tau) $, and
$MGU(\mathcal{F}^{\Theta}_Z (\sigma_1) , \mathcal{F}^{\Theta}_Z (\tau) )= (MGU(\sigma_1, \tau ))_{\restriction{TVar \setminus \Theta}}$. Therefore, $\bigoplus_i U_{\pi_i} =\overline{U}_{\restriction{(TVar\setminus \Theta)}}$.
\\ (iii) First of all, notice that $\overline{U}_{\restriction{(TVar\setminus \Theta)}} (\sigma_2)= \overline{U}(\sigma_2)$ (up-to injective renaming of variables). Moreover,
${\mathcal T}_Z({\mathcal O}({\mathcal R}(\sigma_1\multimap \sigma_2)\cdot {\mathcal R}(\tau)))= \bigoplus_i U_{\pi_i} (\sigma_2)$, since output occurrences of different non-symmetric trajectories are not unifiable. Then the thesis follows by item (ii) of this proposition.
\end{proof}

The following lemma amounts to the main result of this section.

\begin{lemma}\label{la}
Let $\sigma_1\multimap \sigma_2$ and $\tau$ be binary types such that $\sigma_1$ and $\tau$ are unifiable with m.g.u. $\overline{U}$, then
\[ {\mathcal R}(\sigma_1 \multimap \sigma_2) \cdot {\mathcal R}(\tau) = {\mathcal R}(\overline{U}(\sigma_2))\ .\]
\end{lemma}
\begin{proof}
By item (iii) of Proposition~\ref{cuore} and Theorem~\ref{clos}.
\end{proof}

\begin{thermo}
For any closed  term $M$  of the affine combinatory logic, we have:  $\pI{M} = {\mathcal R}(\sigma)$, where $\sigma$ is the principal type of $(M)_{\lambda}$.
\end{thermo}
\begin{proof}
By induction on the structure of $M$. For $M$ a base combinator, one can directly check that the partial involution interpreting $M$ coincides with the relation induced by its principal type.
If $M\equiv M_1M_2$, then, by induction hypothesis, $\pI{M_1} = {\mathcal R}(\sigma_1 \multimap \sigma_2)$ and $\pI{M_2} = {\mathcal R}(\tau)$, where $\sigma_1 \multimap \sigma_2$ and $\tau$ are the principal types of $M_1$ and $M_2$, respectively. The thesis follows by Lemma~\ref{la}.
\end{proof}

Finally, from Proposition~\ref{p6}, we have:

\begin{corollary}\label{finale}
For any $M\in {\mathbf \Lambda}_A^0$, we have: $\pI{M} ={\mathcal R}(\sigma)$, where $\sigma$ is the principal type of $M$.
\end{corollary}

As a consequence, we have also:

\begin{proposition}
The partial involutions of the combinatory algebra ${\mathcal I}$ which are denotations of closed $\lambda_A$-terms are those induced by their principal types.
\end{proposition}

This provides an answer to the open problem raised in \cite{Abr05}.

Another consequence of Corollary~\ref{finale} is  that, since principal types induce an affine combinatory algebra which is a $\lambda$-algebra only on the linear fragment of $\lambda$-calculus, the same holds for the algebra of partial involutions.
Hence the algebra of partial involutions  fails to be a $\lambda$-algebra already on the affine fragment, without considering the \emph{replication} operator used for accommodating exponentials, see \cite{Abr05}.

We conclude this section by showing why we need to refer to ancestor types in establishing the exact relationship  between resolution of principal types and GoI application of the corresponding partial involutions. There exist types which do not produce any resolution, but nonetheless when their respective partial involutions are GoI applied, do yield in fact some result. Ancestor types permit to overcome this mismatch. Namely there exist ancestor types of the types which do not resolve, which do in fact resolve and yield exactly the type whose corresponding partial involution is the result of the GoI application of their corresponding involutions.

\begin{example} Consider the following types which yield the empty resolution:
 \[\sigma \equiv ((\alpha \multimap \beta)\multimap (\gamma \multimap(\gamma \multimap \delta)\multimap \delta))\multimap \alpha \multimap \beta \ \ \ \ \ \ \ \ \ \ \
\tau \equiv (\alpha \multimap \alpha)\multimap (\gamma \multimap \gamma).\]
However we have:\\
${\mathcal R}(\sigma) =\{ rlx\! \leftrightarrow\! lllx, rrx\!\leftrightarrow\!llrx, lrlx\!\leftrightarrow\! lrrlx, lrrlrx\!\leftrightarrow\! lrrrx\}$\\
$ \mathcal R ( \tau) = \{ llx\leftrightarrow lrx, rlx\leftrightarrow rrx\}$\\
$ \mathcal R(\sigma)\cdot \mathcal R (\tau)= \{lx\leftrightarrow rx\}$

\noindent The effect of GoI application can be achieved considering suitable \emph{ancestral types} $\sigma'$ of $\sigma$ and $\tau'$ of $\tau$, which yield the resolution $\alpha\multimap \alpha$ as follows:
\[\sigma'\equiv ((\alpha \multimap \beta)\multimap \gamma)\multimap \alpha \multimap \beta\ \ \ \ \ \ \ \ \ \ \ \tau' \equiv (\alpha \multimap \alpha)\multimap \gamma\]
${\mathcal R}(\sigma') =\{ rlx\! \leftrightarrow\! lllx, rrx\!\leftrightarrow\!llrx\}$\\
$ \mathcal R ( \tau') = \{ llx\leftrightarrow lrx\}$\\
$ \mathcal R(\sigma)\cdot \mathcal R (\tau) = \mathcal R(\sigma')\cdot \mathcal R (\tau')= \{lx\leftrightarrow rx\}.$
\end{example}

\subsection{Two notions of unification}\label{twoviews}
In this section we formally state,   what was pointed out also in \eg\ \cite {H23}, namely that GoI application gives rise to a bottom-up {\em variable-occurrence oriented}  characterization of unification alternate to the standard one.  For simplicity, we consider only binary types where each variable occurs exactly twice. We put:
\begin{definition}[GoI-unification] Let $\sigma,\tau\in T_\Sigma$ be types.
The types $\sigma$ and $\tau$ {\em GoI-unify} if
\\ (i) for every  $\langle u[\alpha],v[\alpha]\rangle \in {\mathcal{R}(\sigma)} $
there exists $\langle u'[\gamma],v'[\gamma]\rangle \in {\mathcal R}(\tau)\hat{;}({\mathcal R}(\sigma)\hat{;}{\mathcal R}(\tau))^*$, such that $uw=u'$ and $vw=v'$, and
\\ (ii) for every  $(u[\alpha],v[\alpha])\in {\mathcal{R}(\tau)}$
there exists $\langle u'[\gamma],v'[\gamma]\rangle \in {\mathcal R}(\sigma)\hat{;}({\mathcal R}(\tau)\hat{;}{\mathcal R}(\sigma))^*$, such that $uw=u'$ and $vw=v'$.

\noindent {\em I.e.:}  \[{\mathcal R}(\tau)\widehat{\subseteq}{\mathcal R}(\sigma)\hat{;}({\mathcal R}(\tau)\hat{;}{\mathcal R}(\sigma))^* \ \ \mbox{ and } \ \ {\mathcal R}(\sigma)\widehat{\subseteq}{\mathcal R}(\tau)\hat{;}({\mathcal R}(\sigma)\hat{;}{\mathcal R}(\tau))^*\]
where $\widehat{\subseteq}$ denotes ``inclusion up-to substitution''.
\end{definition}
\begin{proposition}
Let $\sigma,\tau\in T_\Sigma$ be binary types where each variable occurs exactly twice. Then $\sigma,\tau$ unify if and only if $\sigma,\tau$ GoI unify.
\end{proposition}
\begin{proof}
W.l.o.g. we consider only the case of $\sigma$.
For each type variable $\alpha\in \sigma$, let $\sigma[\alpha,\alpha]$ be $\sigma$, where we have highlighted the two occurrences of a variable $\alpha$, $u[\alpha],v[\alpha]$. Consider the new type $\sigma[\alpha_1,\alpha_2]\multimap \alpha_1\multimap \alpha_2$ and compute ${\mathcal R}(\sigma[\alpha_1,\alpha_2]\multimap \alpha_1\multimap \alpha_2)\cdot {\mathcal R}(\tau)$. Then, by Proposition~\ref{cuore} and Theorem~\ref{clos}, $\sigma$ and $\tau$  unify with unifier $U$ if and only if $U(\alpha)={\mathcal T}_Z(S )$, for a suitable set $Z$ of fresh variables, where $S$ is the collection of all possible outcomes $w[\xi_\alpha]$, for some fresh variable $\xi_\alpha$, one for each type variable $\alpha\in \sigma$:
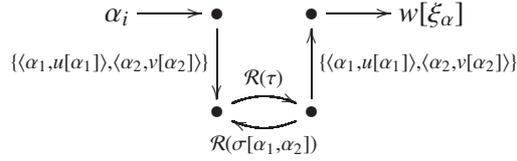
 \[\xymatrix{
\alpha_i \ar[r] & \bullet\ar[d]_{\{\langle \alpha_1,u[\alpha_1]\rangle, \langle \alpha_2,v[\alpha_2]\rangle\}} & \bullet\ar[r] & w[\xi_\alpha]\\
                  & \bullet\ar@/^/[r]^{{\mathcal R}(\tau)} & \bullet\ar@/^/[l]^{{\mathcal R}(\sigma[\alpha_1,\alpha_2])}\ar[u]_{{\{\langle \alpha_1,u[\alpha_1]\rangle, \langle \alpha_2,v[\alpha_2]\rangle\}}} &
}\]
\end{proof}

\section{Final Remarks and Directions for Future Work}\label{finrem}
In this paper, we have established the structural analogy between the interpretation of affine $\lambda$-terms  as partial involutions in a  GoI model {\em \`a la} Abramsky  and their principal types. This allows for understanding  GoI linear application as  resolution, albeit using a variable-directed implementation of unification. We have given full proofs of the equivalence of the finitary type semantics {\em \`a la} Coppo-Dezani (\cite{CDHL82}) and the GoI semantics.

We have argued also that the use of combinators is necessary if we are willing to highlight the correspondence with principal types without having to do any quotienting.

We are confident that this approach extends to full untyped $\lambda$-calculus, as well as to its computational complexity restricted subcalculi, see \cite{Girard1,Girard11,Girard}. To this end,  it is necessary to generalize the type discipline along the lines of \cite{CDGHLS19,HLS20}, using modal  and intersection operators, and to extend the MGU-algorithm to deal with these new constructors. We believe that an appropriate notion of principal type can be introduced also in that case.

As a by-product of our work, we have given also a first answer, for the affine case, to the open problem raised  in \cite{Abr05}, concerning which partial involutions are interpretations of combinatory terms. We are confident that the result that we have obtained will naturally extend to the full $\lambda$-calculus, \ie\ the partial involutions interpreting combinatory terms on the full combinatory algebra are  exactly those structurally isomorphic to their principal types in the general sense.

We believe that the present work sheds more  light on \cite{Girard2}, where the connection between GoI and resolution was originally pointed out in the   context of $C^*$-algebras (see also \cite{Baillot01}).  We think that, following the approach of the present paper and of \cite{DL22}, where  $\lambda$-nets are related to principal types, further connections will arise between principal types and other GoI models, such as \eg\ token machines, categorical semantics, and context semantics. This will contribute to establishing precise connections  between the various GoI models arising in the literature. Another intriguing line of investigation builds on the connection between Levy labels  and types stemming from the seminal work \cite{L78} and further developed in calculi of explicit substitutions, and on the connections between GoI and optimal reductions.

\end{document}